\documentclass[10pt, conference]{IEEEtran}
\usepackage{amsmath,mathrsfs,amssymb,cite,multirow,subfigure,amsthm}
\usepackage{graphicx}
\usepackage{url}
\usepackage{algorithm}
\usepackage{algorithmic}
\usepackage{makecell}

\newtheorem{example}{Example}
\newtheorem{lemma}{Lemma}

\newcommand{\df}{\stackrel{\mbox{\scriptsize def}}{=}}

\renewcommand{\Pi}{P_{\mbox{\tiny{I}}}}

\makeatletter

\newcommand{\Rmnum}[1] {\expandafter\@slowromancap \romannumeral #1@}
\makeatother

\begin{document}
\title{Enhanced Algebraic Error Control for Random Linear Network Coding}
\author{Zhiyuan Yan and Hongmei Xie \\Department of Electrical and Computer
Engineering, Lehigh University \\ Bethlehem, PA 18015, USA
\\
Email: \{yan, hox209\}@lehigh.edu}
\maketitle
\thispagestyle{empty}

\begin{abstract}
Error control is significant to network coding, since when unchecked, errors greatly deteriorate the throughput gains of network coding and seriously undermine both reliability and security of data. Two families of codes, subspace and rank metric codes, have been used to provide error control for random linear network coding. In this paper, we enhance the error correction capability of these two families of codes by using a novel two-tier decoding scheme. While the decoding of subspace and rank metric codes serves a second-tier decoding, we propose to perform a first-tier decoding on the packet level by taking advantage of Hamming distance properties of subspace and rank metric codes. This packet-level decoding can also be implemented by intermediate nodes to reduce error propagation. To support the first-tier decoding, we also investigate Hamming distance properties of three important families of subspace and rank metric codes, Gabidulin codes, K\"{o}tter--Kschischang codes, and Mahdavifar--Vardy codes. Both the two-tier decoding scheme and the Hamming distance properties of these codes are novel to the best of our knowledge.
\end{abstract}


\section{Introduction}\label{sec: intro}
Network coding \cite{ahlswede_IT00} has the potential to fundamentally transform current and future communication networks (CNs) due to its promise of significant throughput gains. However, a key obstacle to its adoption in practical CNs is its vulnerability to errors caused by unreliable links or malicious nodes. If unchecked, errors greatly deteriorate the throughput gains of network coding and seriously undermine both reliability and security of data.

In this paper,  we focus on algebraic error control for network coding, coding-theoretic end-to-end error correction for network coding \cite{cai_itw02}. Similar to classical error control coding, end-to-end error correction for network coding involves only two ends of a  \emph{multicast}: an encoder at the source node adds redundancy to the transmitted data so that the decoder at any destination node can distill out data in the presence of errors, while the intermediate nodes are oblivious. In contrast, other error control approaches for network coding --- link-level error control (or \emph{channel coding}), packet-level error control such as cyclic redundancy check, and cryptographic approach
 --- all require extra operations at every intermediate node in the network. End-to-end error correction is embedded in network coding and does not require additional infrastructure as some cryptographic schemes do.

Also, we focus on error control for random linear network coding (RLNC) \cite{HMKKESL_IT06}.
In RLNC, all packets are treated as vectors over some finite field, or \emph{Galois field} (GF) of size $q$, denoted by GF$(q)$.
Given incoming packets
$u_1, u_2, \cdots, u_n$, an intermediate node forms an outgoing packet $v_i$ by linearly combining these vectors, that is, $v_i=\sum_{j=1}^{n} a_{i,j}u_j$, where $a_{i,j}$'s are randomly chosen from GF$(q)$.
Instead of using  network coding operations centrally designed to achieve the maximum throughput based on network topologies, RLNC achieves the maximum throughput \cite{HMKKESL_IT06} despite its distributed and random nature. Hence, RLNC is ideal for CNs that either are decentralized or have time-varying topologies \cite{HMKKESL_IT06}.
Our work focuses on RLNC due to its significance, but can be readily extended to more general network coding schemes.

Algebraic error control proposed for RLNC assume either \emph{coherent} or \emph{noncoherent} transmission models. Error control schemes of the first type \cite{cai_itw02} depend
on and take advantage of the underlying network topology or the particular network coding operations performed at intermediate
network nodes. Error control schemes of the second type \cite{KK_IT0808} assume
that the transmitter and receiver have no knowledge of such channel
transfer characteristics.

Two families of codes, \emph{subspace codes} \cite{KK_IT0808} and \emph{rank metric codes} \cite{Delsarte,Gab85, Roth91}, are appropriate codes for error control in noncoherent and coherent network coding, respectively. For subspace and rank metric codes, the relevant metrics for error control are not Hamming metrics. Also, their decoders operate on a set of packets at the destination nodes of a multicast. The first main contribution of this paper is an enhanced control scheme that involves two-tier decoding. In our error control scheme, the decoding of subspace or rank metric codes is the second-tier decoding, and a first-tier decoding is carried out on a packet level. This enhanced error control is enabled by the fact that for each data session a valid packet (a packet that is a linear combination of the transmitted packets) belongs to a linear block code. Thus, a subspace (or rank metric) code correspond to a collection of linear block codes. By taking advantage of the Hamming distance properties of the collection of linear block codes corresponding to subspace or rank metric codes, the first-tier decoding can enhance the error correction capability of subspace and rank metric codes. Since the first-tier decoding is on the packet level, it can be implemented by intermediate nodes to reduce error propagation. In addition to the two-tier decoding scheme, the other main contribution of this paper is the Hamming distance properties of the collection of linear block codes corresponding to three important families of subspace and rank metric codes, Gabidulin codes, K\"{o}tter--Kschischang (KK) codes, and Mahdavifar--Vardy (MV) codes. Both the two-tier decoding scheme and the Hamming distance properties of these codes are novel to the best of our knowledge.

The rest of the paper is organized as follows. Some preliminaries on subspace and rank metric codes are introduced in Section~\ref{sec: pre}. Section~\ref{sec:twotier} proposes our two-tier decoding scheme for subspace and rank metric codes. In Section~\ref{sec:Hamming}, we investigate the Hamming distance properties of Gabidulin, KK and MV codes.
Section~\ref{sec:conclusions} concludes the paper.

\section{Preliminary}\label{sec: pre}
\subsection{Error Control for Random Linear Network Coding}\label{sec:background}
Two families of codes, \emph{subspace codes} \cite{KK_IT0808} and \emph{rank metric codes} \cite{Delsarte,
Gab85, Roth91}, are appropriate codes for error control in noncoherent and coherent network coding, respectively.
A subspace code is a subset of the projective space \cite{KK_IT0808}. Two metrics, the subspace metric \cite{KK_IT0808} and the injection metric \cite{SK_IT1209}, have been defined for subspace codes.
When all subspaces over the operator channel are of the same dimension, a subspace code is reduced to a \emph{constant-dimension code} (CDC), a subset of all subspaces with the same dimension (called a \emph{Grassmannian}).
CDCs are interesting since they lead to simplified network protocols due to the fixed dimension.
Rank metric codes are important for two reasons. First, error control for coherent transmission models can be solved by using rank metric codes \cite{SK_IT1209}. In particular, Gabidulin codes \cite{Gab85}, a class of rank metric codes optimal with respect to
the Singleton bound \cite{Delsarte,
Gab85, Roth91}, maximize the error correction capability in coherent network coding \cite{SK_IT1209}.
Second, rank metric codes provide an alternative approach to investigating subspace codes, because subspace codes are intricately related to rank metric codes. The K\"{o}tter--Kschischang (KK) codes \cite{KK_IT0808}, a class of CDCs, are motivated by and
related to Gabidulin codes via the lifting operation \cite[Definition~3]{SKK_IT0908}.  Mahdavifar and Vardy proposed a family of CDCs with improved error correction capability \cite{mahdavifar_isit10,mahdavifar_it10}.

\subsection{Gabidulin Codes}\label{sec: GB}
Delsarte~\cite{Delsarte}, Gabidulin~\cite{Gab85} and Roth~\cite{Roth91} did pioneering work on rank metric codes. The rank distance between two vectors $\mathbf{x, y} \in$ GF$(q^m)^n$ is defined to be $d_r(\mathbf{x}, \mathbf{y}) = r(\mathbf{x-y}; q)$, where $r(\mathbf{x}; q)$ is the rank of vector $\mathbf{x}$ over GF$(q)$. The minimum rank distance of a code $\mathcal{C}$, denoted as  $d_r(\mathcal{C})$, is simply the minimum rank distance over all possible pairs of distinct codewords. For linear $(n,k)$ codes over GF$(q^m)$, when $n\le m$, the Singleton bound gives $d_r(\mathcal{C}) \le n-k+1$, and codes achieving the equality are called maximum-rank-distance (MRD) codes.

Gabidulin codes are a family of MRD codes proposed by Gabidulin \cite{Gab85}. An $(n,k)$ Gabidulin code $\mathcal{C_G}$ over GF$(q^m)$ ($n \le m$) is generated by $n$ elements $g_0, g_1, \ldots, g_{n-1}\in$ GF$(q^m)$ that are linearly independent over GF$(q)$. Given a message vector $\mathbf{u} = (u_0, u_1, \ldots, u_{k-1})\in$ GF$(q^m)^k$, the corresponding codeword in $\mathcal{C_G}$ is $\mathbf{c} = (u(g_0), u(g_1), \ldots, u(g_{n-1}))^T$, where $u(x)=\sum_{i=0}^{k-1}u_i x^{[i]}$ is a linearized polynomial with $[i]\df q^i$.

\subsection{KK Codes}
\label{sec: introKK}
KK codes~\cite{KK_IT0808} are a type of subspace codes for random linear network coding, where each codeword is a linear subspace of some ambient space $W$.
A KK code $\mathcal{C_K}$ is defined by $l$ ($l \le m$) elements $\alpha_0, \alpha_1, \ldots, \alpha_{l - 1} \in$ GF$(q^m)$ that are linearly independent over GF$(q)$. Given a message vector $\mathbf{u} = (u_0, u_1, \ldots, u_{k-1})\in$ GF$(q^m)^k$, a linearized polynomial is formed by $u(x) = \sum_{i = 0}^{k-1} u_i x^{[i]}$. Then a codeword of $\mathcal{C_K}$ is given by the $l$-dimensional subspace of $W = \langle \alpha_0, \alpha_1, \ldots, \alpha_{l - 1} \rangle \oplus$ GF$(q^m)$ spanned by $\{(\alpha_i, \beta_i): \beta_i = u(\alpha_i), i = 0, 1, \ldots, l-1\}$.


\subsection{MV Codes}
\label{sec: introMV}
MV codes~\cite{mahdavifar_isit10} are similar to but different from KK codes~\cite{KK_IT0808}. Suppose $\mathcal{C}_{\textrm{MV}}$ is an $l$-dimensional MV code over GF$(q^{ml})$ ($l$ divides $q-1$), generated by $\alpha_0, \alpha_1, \ldots, \alpha_{l-1}\in$ GF$(q^{ml})$, a specially chosen set of elements linearly independent over GF$(q)$. Then message vectors $\mathbf{u} = (u_0, u_1, \ldots, u_{k-1})$ are defined over GF$(q)$, and a linearized polynomial is formed by $u(x) = \sum_{i=0}^{k-1} u_i x^{[i]}$. Let $u^{\otimes i}(x)$ denote the composition of $u(x)$ with itself by $i$ times for any nonnegative integer $i$, and $u^{\otimes 0 }(x) =x$. Then the codeword $V$ corresponding to the message $\mathbf{u}$ is spanned by a set of vectors $v_i$ for $i = 1,2,\ldots, l$, where $v_1 = (\alpha_1, u(\alpha_1), u^{\otimes 2}(\alpha_1), \ldots, u^{\otimes L}(\alpha_1))$, $v_i = (\alpha_i, \frac{u(\alpha_i)}{\alpha_i},\ldots, \frac{u^{\otimes L}(\alpha_i)}{\alpha_i})$, and $L$ is some list size desired at the decoder. Note that $\frac{u^{\otimes j}(\alpha_i)}{\alpha_i} \in$ GF$(q^m)$ for any $j \ge 0$ and $i=2,3,\ldots, l$~\cite{mahdavifar_it10}. Then $V$ is an $l$-dimensional subspace of the $(l+Lm)$-dimensional ambient space $W = \langle \alpha_1, \alpha_2,\ldots, \alpha_l \rangle \oplus \underbrace{ \textrm{GF}(q^m) \oplus \cdots \oplus \textrm{GF}(q^m)}_{L\textrm{ times}}$.

\section{Two-Tier Decoding of Subspace and Rank Metric Codes}\label{sec:twotier}
The two-tier decoding is enabled by a key observation that all valid packets for existing subspace and rank metric codes constitute a collection of linear block codes.  For simplicity, we will assume a random linear network coding over GF$(q)$ and illustrate this with a CDC, which is a set of subspaces $V_i$ (with dimension $l$) of GF$(q)^m$ ($l\leq m$). We denote the basis of $V_i$ as $\left\{v_{i,0}, v_{i,1}, \cdots, v_{i, l-1}\right\}$, where $v_{i,j} \in \mbox{GF}(q)^m$ is a row vector. Thus, when $V_i$ is selected for the multicast, its basis $\left\{v_{i,0}, v_{i,1}, \cdots, v_{i, l-1}\right\}$ is injected into the network. At any destination node, a valid packet is of the form
$\sum_{j=0}^{l-1}a_j v_{i,j} = (a_0 \, a_1 \, \cdots a_{l-1}){\mathbf v}_i$, where $a_j\in$GF$(q)$ and ${\mathbf v}_i=\left[v_{i,0}^T\, v_{i,1}^T\, \cdots \, v_{i, l-1}^T\right]^T$. That is, all valid packets for $V_i$ constitutes a linear block code over  GF$(q)$ with a generater matrix ${\mathbf v}_i$, denoted as $C_i$, and all valid packets corresponding to all subspaces constitute a union code ${\mathcal{C_U}} = \bigcup_{V_i} {C}_i$. We note that ${\mathcal{C_U}}$ is not necessarily a linear code. Also, ${\mathcal{C_U}}$ depends on the CDC used, not the transmitted data. Hence, it can be assumed that ${\mathcal{C_U}}$ is known to all nodes.

At a destination node of a multicast, our two-tier decoding scheme is as follows. First, since the union code ${\mathcal{C_U}}$ is known, the destination node performs error detection or correction on each received packet by taking advantage of the Hamming distance properties of ${\mathcal{C_U}}$. For instance, when a received packet does not belong to the union code, it can be discarded. Alternatively, the minimum Hamming distance of ${\mathcal{C_U}}$ also ensures a certain correction radius. Thus, if the number of bit errors in a received packet is within the error correction radius, the first-tier decoding can also correct the bit errors. The second-tier decoding is performed on the set of packets produced by the first-tier decoder. By removing packets that are invalid and correcting packets corrupted by few bit errors, the first-tier decoder improves the error correction capability of the second-tier decoder.

By allowing the first- and second-tier decoders to pass information to each other, both can be further enhanced.
First, we show that how information from the first-tier decoder can help the second-tier decoder.  When errors have known locations, they are called \emph{erasures}. As in classical coding theory, the generalized rank decoder in \cite{SKK_IT0908} can correct \textbf{twice} as many erasures as errors by taking advantage of the extra location knowledge.
Thus, when the first-tier decoder marks unreliable packets as erasures for the second-tier decoder, the first-tier decoder enhances the correction capability of the second-tier decoder. Second, information from the second-tier decoder can be used to help the first-tier decoder. For instance, if the second-tier decoder is a list decoder (such as those proposed by \cite{mahdavifar_it10}), the second-tier decoder outputs a list ${\mathcal L}$ of subspaces as possible transmitted subspaces. Given this information, all valid packets now constitute $\bigcup_{V_i \in {\mathcal L}}  C_i$. Given the list ${\mathcal L}$ from the second-tier decoder, the first-tier decoder can re-decode the received packets with respect to $\bigcup_{V_i \in {\mathcal L}} {C}_i$. Since $\bigcup_{V_i \in {\mathcal L}} { C}_i \subseteq \bigcup_{V_i} {\mathcal{C_U}}$, this information from the second-tier decoder will improve the error detection/correction capability of the first-tier decoder.

We note that our two-tier decoding scheme described above is carried out on the packet level and at the destination nodes only, without affecting intermediate nodes. However, the first-tier decoding can be implemented by intermediate nodes involved in the multicast. This is because ${\mathcal{C_U}}$ and its Hamming distance properties are known to all nodes. In contrast, the second-tier decoding, existing decoders for subspace and rank metric codes, cannot be performed at intermediate nodes. For instance, an intermediate node may have only one packet in a session, and cannot perform the second-tier decoding. If intermediate nodes can discard invalid packets or correct packets corrupted by few bit errors, error propagation can be reduced at the expense of additional complexities at intermediate nodes.

Finally, we provide two remarks regarding the first-tier decoder. First, we emphasize that the first-tier decoding takes advantage of Hamming distance properties of  $\bigcup_{V_i} {C}_i$, and hence is different from link-level channel coding, cyclic redundancy check, and the cryptographic operations. Second, the improved error correction capability from using first-tier decoding does not require any additional redundancy, since the first-tier decoding takes advantage of the redundancy that already exists in the Hamming metric space.

\section{Hamming Distance Properties of Subspace and Rank Metric Codes}\label{sec:Hamming}
As discussed above, Hamming distance properties of  ${\mathcal{C_U}}$ are important to our two-tier decoding scheme. Thus, we investigate the Hamming distance properties of  ${\mathcal{C_U}}$ corresponding to three important families of subspace and rank metric codes, Gabidulin codes, KK codes, and MV codes. Our work in this area revolves around two aspects. The first is the minimum Hamming distance of $\bigcup_{V_i} {C}_i$, which is vital to the first-tier decoder. Furthermore, we also investigate the minimum Hamming distances of the individual component codes ${C}_i$'s for two reasons. First, the minimum Hamming distances of the individual component codes ${C}_i$'s provide an upper bound on the minimum Hamming distance of ${\mathcal{C_U}}$. Second, as described above, when partial information about the transmitted subspace is available, the first-tier decoder may consider only a small set ${\mathcal L}$ of component codes, $\bigcup_{V_i \in {\mathcal L}} {C}_i$. The minimum Hamming distances of ${C}_i$ also help us determine the minimum distance of $\bigcup_{V_i \in {\mathcal L}} {C}_i$. Also, even when the minimum Hamming distance of ${\mathcal{C_U}}$ is one, the first-tier decoding is possible as long as ${\mathcal{C_U}}$ is not the whole ambient space. Finally, we assume a random linear network coding over GF$(q)$, and consider Gabidulin, KK, and MV codes over extension fields of GF$(q)$. We note that
henceforth in this paper, an element in an extension field GF$(q^m)$ is sometimes treated as a length-$m$ row vector over GF$(q)$, depending on the context.

\subsection{Hamming Distance Properties of Gabidulin Codes}
\label{sec: HD_GB}
Consider an $(n, k)$ Gabidulin code $\mathcal{C_G}$ over GF$(q^m)$, generated by $g_0, g_1, \ldots, g_{n-1}\in$ GF$(q^m)$. For a codeword $\mathbf{c}\in$ $\mathcal{C_G}$, we treat each encoded symbol $u(g_i)$ as an $m$-dimensional row vector over GF$(q)$, and obtain an $n\times m$ matrix $G$. Using $G$ as a generator matrix, we obtain a linear block code $C$, can call codewords of $C$ \emph{valid vectors}. The union of all valid vectors corresponding to all the codewords of $\mathcal{C_G}$ is called a \emph{union code} $\mathcal{C_U}$, and each linear block code $C$ is referred to as a \emph{component code} of $\mathcal{C_U}$. We want to find the minimum Hamming distance the union code $\mathcal{C_U}$ and its component codes.


Any $(n,k)$ Gabidulin code $\mathcal{C_G}$ contains codewords $\mathbf{c}$ corresponding to $u(x)=u_i x^{[i]}$ with $u_i$ an arbitrary element in GF$(q^m)$ for $0\le i \le k-1$, while $u_{i'}=0$ for $i'=0,1,\ldots, i-1,i+1,\ldots,k-1$. Hence codewords of the component code $C$ generated by $u_ig_0^{[i]}, u_ig_1^{[i]}, \ldots, u_ig_{n-1}^{[i]}$ are always valid vectors of the union code $\mathcal{C_U}$. The distance property of those valid vectors reflects the minimum distance of $\mathcal{C_U}$.
\begin{lemma}\label{lemma: HD_GB_U}
The minimum Hamming distance of the union code $\mathcal{C_U}$ is 1.
\end{lemma}
\begin{proof}
For a nonzero $u_i\in$ GF$(q^m)$, we have $u_i = \alpha^j$ for some $j \in \{0,1,\ldots, q^m-2\}$, where $\alpha$ is a primitive element in GF$(q^m)$. There are $q^m-1$ component codes $C_j$'s generated by generator matrices
\begin{equation} \label{equ: GB_sub_G}
G_j = \left( \begin{array}{c} \alpha^j g_0^{[i]} \\ \alpha^j g_1^{[i]} \\ \vdots \\ \alpha^j g_{n-1}^{[i]} \end{array} \right),
\end{equation}
for $j=0,1,\ldots, q^m-2$, respectively. In particular, the first row in $G_j$, $\alpha^j g_0^{[i]}$, is a valid vector of $\mathcal{C_U}$. Since $g_0$ is a nonzero element over GF$(q^m)$, so is $g_0^{[i]}$, hence we can express it as $g_0^{[i]} = \alpha^A$ for some $A \in \{0,1,\ldots,q^m-2\}$. Valid vectors include $\alpha^j \alpha^A=\alpha^{A+j}$ for $j=0,1,\ldots,q^m-2$, which are exactly all the nonzero elements in GF$(q^m)$, leading to a minimum Hamming distance of 1 for the union code.
\end{proof}


Now we consider minimum distances of the component codes. The general distance property is presented first, followed by some special cases on a code-to-code basis.
\begin{lemma}\label{lemma: HD_GB_sub}
Given an $(n,k)$ Gabidulin code $\mathcal{C_G}$, the minimum Hamming distance $d_H(C)$ of each component code $C$ satisfies $d_H(C)\le m-n+k$.
\end{lemma}
\begin{proof}
Suppose $\mathbf{c}$ is a codeword of $\mathcal{C_G}$, and the corresponding component code of $\mathcal{C_U}$ is $C$. Then the dimension of $C$, denoted by $k'$, is exactly $r$, the rank of the codeword $\mathbf{c}$ in $\mathcal{C_G}$. The Singleton bound gives $d_H(C) \le m - k'+1=m-r+1$. Since Gabidulin codes are MRD codes, $r\ge d_r(\mathcal{C_G})=n-k+1$. Hence we have $d_H(C)\le m-n+k$.
\end{proof}


\begin{lemma} \label{lemma: HD_GB_sub1}
The component code $C_j$ generated by Eq.~(\ref{equ: GB_sub_G}) has a minimum Hamming distance of $d_H(C_j)\le m-n+1$ for any $j\in\{0,1,\ldots, q^m-2\}$.
\end{lemma}
\begin{proof}
Since $g_0, g_1,\ldots, g_{n-1}$ are linearly independent over GF$(q)$, $\alpha^j g_0^{[i]}, \alpha^j g_1^{[i]},\ldots, \alpha^j g_{n-1}^{[i]}$ are also linearly independent for fixed $i, j \in\{0,1,\ldots, q^m-2\}$. Otherwise, we can find a set of nontrivial elements $a_0, a_1,\ldots,a_{n-1}\in$ GF$(q)$, such that $\sum_{s=0}^{n-1} a_s (\alpha^j g_{s}^{[i]})=\alpha^j \sum_{s=0}^{n-1} a_s g_{s}^{[i]}=0$. Note that $\alpha^j$ is a nonzero element over GF$(q^m)$, hence we can multiply the previous equation with $(\alpha^i)^{-1}$ on both equation, and obtain $\sum_{s=0}^{n-1} a_s g_{s}^{[i]}= (\sum_{s=0}^{n-1} a_s g_{s})^{[i]}=0$, leading to $\sum_{s=0}^{n-1} a_s g_{s}=0$. This contradicts the assumption that $g_0, g_1,\ldots, g_{n-1}$ are linearly independent. Hence the generator matrix $G_j$ has a full rank $n$, hence generating an $n$ dimensional linear block code with length $m$, resulting in $d_H(C_j)\le m-n+1$ from Singleton bound.
\end{proof}

From Lemma~\ref{lemma: HD_GB_sub1}, when $n=m$, all component codes generated by $G_j$ in Eq.~(\ref{equ: GB_sub_G}) have minimum Hamming distance of 1 since $G_j$ spans the entire space of GF$(q^m)$. When a polynomial basis is used to represent the elements of GF$(q^m)$, component codes with a minimum Hamming distance of 1 exist even if $n< m$. Let us consider $G_j$ in Eq.~(\ref{equ: GB_sub_G}) with $i=0$. From previous analysis, we know that for a fixed $g_s$ with $s\in\{0,1,\ldots,n-1\}$, there exists $A_s \in \{0,1,\ldots,q^m-2\}$ such that there exists a valid vector $\alpha^{A_s} g_s = 1$. Since $g_s$'s are different, $A_s$'s are also different, leading to at least $n$ component codes with minimum Hamming distance of 1 when a polynomial basis is used.


\subsection{Hamming Distance Properties of KK Codes}
\label{sec: HD_KK}
Consider an $l$-dimensional KK code $\mathcal{C_{KK}}$ over GF$(q^m)$, generated by $l$ linearly independent elements $\alpha_0, \alpha_1, \ldots, \alpha_{l-1}\in$ GF$(q^m)$. Each codeword or subspace $C$ of the KK code $\mathcal{C_{KK}}$ is an $l$-dimensional subspace, called a component code, and vectors contained in $C$ valid vectors. The union of valid vectors contained in all the codewords of $\mathcal{C_{KK}}$ is referred to as a union code $\mathcal{C_U}$. Note that each valid vector can be written as $(a,b)$, where $a\in$ GF$(q^m)$ is a linear combination of $\alpha_0, \alpha_1, \ldots, \alpha_{l-1}$, and $b\in$ GF$(q^m)$ is obtained by $b=u(a)$, where $u(x)$ is the linearized polynomial from the message vector.

We also start from the minimum Hamming distance of the union code. Following similar arguments in Section~\ref{sec: HD_GB}, we consider valid vectors obtained from linearized polynomials $u(x)=u_i x^{[i]}$ with $u_i\in$ GF$(q^m)$ for $i\in \{0,1,\ldots, k-1\}$.

\begin{lemma} \label{lemma: HD_KK_U}
The minimum Hamming distance of the union code $\mathcal{C_U}$ is 1.
\end{lemma}
\begin{proof}
For a nonzero $u_i\in$ GF$(q^m)$, we express $u_i=\gamma^j$ for some primitive elements $\gamma \in$ GF$(q^m)$ and some $j\in\{0,1,\ldots, q^m-2\}$. Consider a component code $C_j$ obtained from the a generator matrix
\begin{equation} \label{equ: KK_sub_Gj}
G_j=\left(\begin{array}{cc}
(\alpha_0, & \gamma^j \alpha_0^{[i]}) \\
(\alpha_1, & \gamma^j \alpha_1^{[i]}) \\
\vdots & \vdots \\
(\alpha_{l-1}, & \gamma^j \alpha_{l-1}^{[i]})
\end{array} \right).
\end{equation}
Clearly, the first row $(\alpha_0,  \gamma^j \alpha_0^{[i]})$ in $G_j$ is a valid vector. Similar to Section~\ref{sec: HD_GB}, we can write the nonzero element $\alpha_0^{[i]} = \gamma^A$ for some $A\in\{0,1,\ldots,q^m-2\}$, and obtain valid vectors $(\alpha_0, b)$, where $b$ can take all the nonzero elements in GF$(q^m)$. Hence the minimum Hamming distance of the union code is 1.
\end{proof}

The minimum distance of the union code $\mathcal{C_U}$ is examined across component codes in Lemma~\ref{lemma: HD_KK_U}. Now for each component code $C$, which is an $l$-dimensional linear block code with length $2m$, the minimum distance satisfies $d_H(C) \le 2m-l+1$. Further, we can construct $\mathcal{C_K}$ such that the minimum distances of component codes are bounded from below.
\begin{lemma} \label{lemma: HD_KK_sub}
Let us denote by $C_0$ the component code corresponding to $u(x)=0$, and $C$ any component code of $\mathcal{C_U}$. Then $d_H(C_0) \le d_H(C)$. Furthermore, we can construct $C_0$ with $d_H(C_0)=m-l+1$ when $q\ge m$.
\end{lemma}
\begin{proof}
A component code $C$ of $\mathcal{C_U}$ is spanned by $(\alpha_s, u(\alpha_s))$'s for $s=0,1,\ldots,l-1$. Hence valid vectors can be written as $(a,b)$, where $a, b\in$ GF$(q^m)$, and $a$ is a linear combination of $(\alpha_s, u(\alpha_s)$'s. In particular, valid codewords of $C_0$, which corresponds to the zero linearized polynomial, always take the form $(a,0)$. Hence given a valid vector $(a,b)$ of component code $C$, there's always a valid vector $(a,0)$ of $C_0$. Note that the Hamming weight of $(a,0)$ is always no greater than $(a,b)$. Hence we have $d_H(C_0) \le d_H(C)$. On the other hand, the minimum distance of $C_0$ is exactly the minimum Hamming distance of a traditional $l$-dimensional linear block code $C_0'$, whose $l\times m$ generator matrix is composed of the $l$ row vectors $\alpha_0, \alpha_1,\ldots, \alpha_{l-1}$, leading to $d_H(C_0)\le m-l+1$. When $q\ge m$, the bound is achievable by selecting $\alpha_0, \alpha_1,\ldots,\alpha_{l-1}$ to be row vectors of the generator matrix of an $(m,l)$ RS codes.
\end{proof}

Lemma~\ref{lemma: HD_KK_sub} points out that $C_0$ has the smallest minimum Hamming distance, hence we want to improve the distance properties of the component code by designing $C_0$ with larger minimum Hamming distance, and Lemma~\ref{lemma: HD_KK_sub} indicates achievability when $q\ge m$. When $q < m$, we can also achieve $d_H(C_0)>1$ by using codes such as BCH codes.

Although Lemma~\ref{lemma: HD_KK_U} establishes a minimum Hamming distance of one for the union code, it is still able to provide error detection based on the union code. The reason is that  the cardinality of the union code is always smaller than that of the ambient space, as shown below in Lemma~\ref{lemma: HD_KK_cardi}.

\begin{lemma} \label{lemma: HD_KK_cardi}
The cardinality of the union code corresponding to a KK code is always smaller than that of the ambient space.
\end{lemma}
\begin{proof}
Suppose $(a,b)$ is a valid vector of the union code, than $a\in$ GF$(q^m)$ is a linear combination of $\alpha_0,\alpha_1,\ldots, \alpha_{l-1}$, and $b=u(a)\in$ GF$(q^m)$ for some message polynomial $u(x)$. In fact, if $a\ne 0$, any $b\in$ GF$(q^m)$ would make $(a,b)$ a valid vector following a similar argument as in Lemma~\ref{lemma: HD_KK_U}. Given that $a$ belongs to the $l$-dimensional subspace spanned by $\alpha_0,\alpha_1,\ldots, \alpha_{l-1}$, there are a total of $q^l-1$ nonzero combinations leading to a valid $a$. Hence there are a total of $(q^l-1)q^m+1=|\mathcal{C_U}|$ valid vectors in the union code, including the all zero element. The ambient space has a dimension of $l+m$ as shown in Section~\ref{sec: introKK}, and hence a total number of $q^{l+m}$ vectors. It is easy to see that $|\mathcal{C_U}|=q^{l+m}-q^m+1 < q^{l+m}$.
\end{proof}

Note that received vectors with a form of $(a,b)$ belong to the $2m$-dimensional space over GF$(q)$. Based on Lemma~\ref{lemma: HD_KK_cardi}, even given the worst case with $l=m$, where the ambient space is the $2m$-dimensional subspace over GF$(q)$, the union code still has error detection capabilities since its cardinality is smaller than $q^{2m}$. Hence it is possible that the union code may provide some error correction capabilities at the receiver when $l < m$.

\begin{example}
We construct a KK code over GF$(2^3)$ with $l=2$ and $k=1$. Select $\alpha_0=\gamma^3$ and $\alpha_1=\gamma^4$, where $\gamma$ is a root of the irreducible polynomial $x^3+x+1$, and can be verified to be a primitive element over GF$(2^3)$. It can be verified that codewords of $C_0$ are $(\gamma^3,0), (\gamma^4,0), (\gamma^6,0)$ as well as the all zero vector, with $d_H(C_0)=2$. Thus, $C_0$ maximizes its minimum Hamming distance. In the union code, component codes corresponding to $u(x)=\gamma^j$ with $j=0,1,\ldots,6$ gives valid vectors $(\gamma^3,\gamma^3), (\gamma^3,\gamma^4), (\gamma^3,\gamma^5), (\gamma^3,\gamma^6), (\gamma^3,1), (\gamma^3,\gamma), (\gamma^3,\gamma^2)$, resulting in $d_H(\mathcal{C_U})=1$.
\end{example}


\subsection{Hamming Distance Properties of MV Codes}
\label{sec: HD_MV}
In this section, we examine the Hamming distances of the component codes first, and then present distance property for the union code $\mathcal{C_U}$.

\begin{lemma} \label{lemma: HD_MV_sub}
Let us denote by $C_0$ the component code corresponding to $u(x)=0$, and $C$ any component code of $\mathcal{C_U}$. Then $d_H(C_0) \le d_H(C)$. Further, we can construct $C_0$ with $d_H(C_0)=ml-l+1$ when $q\ge ml$.
\end{lemma}
\begin{proof}
The proof follows similar arguments in KK codes case as stated in Lemma~\ref{lemma: HD_KK_sub}. Note that the generator matrix of $C_0$ is
\begin{displaymath}
G_0 = \left( \begin{array}{ccccc}
\alpha_0 & 0 & 0 & \ldots & 0 \\
\alpha_1 & 0 & 0 & \ldots & 0 \\
\vdots & \vdots & \vdots & \ddots & \vdots\\
\alpha_{l-1} & 0 & 0 & \ldots & 0
\end{array} \right).
\end{displaymath}
Accordingly, $d_H(C_0)$ is exactly the same as $d_H(C_0')$, where $C_0'$ is an $(ml,l)$ linear block code whose $l\times ml$ generator matrix is composed of $\alpha_0, \alpha_1,\ldots, \alpha_{l-1}$, resulting in $d_H(C_0)\le ml-l+1$. Further, when $q\ge ml$, the equality is achievable by selecting $\alpha_0, \alpha_1,\ldots,\alpha_{l-1}$ to be row vectors of the generator matrix of an $(ml,l)$ RS codes.
\end{proof}
Note when $q< ml$, we can still use codes such as BCH codes to construct $C_0$ with $d_H(C_0)>1$ such that the minimum Hamming distance of each component code is bounded from below.

\begin{lemma} \label{lemma: HD_MV}
When $l=1$, $d_H(\mathcal{C_U}) \le m$, and $d_H(\mathcal{C_U}) \le \textrm{ min }\{ml-l+1, L\}$ if $l>1$. \end{lemma}
\begin{proof}
Again we start from the case with $u(x)=u_i x^{[i]}$, where $u_i \in$ GF$(q)$. If $l=1$, we obtain one-dimensional linear block codes generated by $(\alpha_0, u_i \alpha_0^{[i]}, u_i^2 \alpha_0^{[2i]},\ldots, u_i^{L} \alpha_0^{[Li]})$, where $L$ is the list size. Also denote the component code corresponding to $u(x)=0$ by $C_0$, and we have $d_H(\mathcal{C_U}) \le d_H(C_0)$. Note that $\alpha_0 \in$ GF$(q^m)$, hence if a polynomial basis is used, we have $d_H(\mathcal{C_U}) \le d_H(C_0) \le m$. Otherwise, a bigger upper bound could be possible. When $l > 1$, based on a fixed $\alpha_s \in$ GF$(q^{ml})$ with $s\in\{1,2,\ldots,l-1\}$, there exist codewords $(\alpha_s, 0,0,\ldots,0)$ corresponding to $u_0=0$, and $(\alpha_s, 1,1,\ldots,1)$ corresponding to $u_0=1$. Hence the minimum distance of the union code is bounded by $d_H(\mathcal{C_U})\le L$ if a polynomial basis is adopted.
Combining Lemma~\ref{lemma: HD_MV_sub}, we have $d_H(\mathcal{C_U}) \le \textrm{ min }\{ml-l+1, L\}$ if $l > 1$, and $d_H(\mathcal{C_U}) \le m$ if $l=1$.
\end{proof}

\begin{lemma} \label{lemma: HD_MV_cardi}
The cardinality of the union code corresponding to an MV code is smaller than that of the ambient space.
\end{lemma}
\begin{proof}
The proof follows a similar argument as in Lemma~\ref{lemma: HD_KK_cardi}. In this case, the cardinality of the union code is at most $(q^l-1)q^{Lm}+1$, and is less than $q^{l+Lm}$, which is the size of the ambient space.
\end{proof}

\begin{example} \label{example:MV1}
We construct an MV code over GF$(2^3)$ with $l=1, k=1$ and $L=2$. Let $\alpha_0=\gamma^5$, where $\gamma$ is a root of the irreducible polynomial $x^3+x+1$ such that $\gamma^5$ and its conjugates form a normal basis of GF$(2^3)$ over GF$(2)$. Then the union code $\mathcal{C_U}$ contains two component codes, $C_0$ generated by $(\gamma^5, 0,0)$, and $C_1$ generated by $(\gamma^5,\gamma^5,\gamma^5)$. We can directly get $d_H(C_0)=3, d_H(C_1)=9$, and $d_H(\mathcal{C_U})=3$.
\end{example}

\begin{example} \label{example:MV2}
Consider another MV code over GF$(3^6)$ with $m=3, l=2, k = 1$, and $L=5$. Let $\alpha_0=\gamma^{504}$ and $\alpha_0=\gamma^{294}$, where $\gamma$ is a root of the irreducible polynomial $x^6+x+2$ such that $\gamma^{15}$ and its conjugates form a normal basis of GF$(3^6)$ over GF$(3)$. It can be verified that $d_H(\mathcal{C_U})=d_H(C_0)=3$, while $d_H(C) = 9$ for any other component code $C$.
\end{example}

Example~\ref{example:MV1} shows a minimum distance equal to $m$ for an MV code with $l=1$, hence the upper bound in Lemma~\ref{lemma: HD_MV} is reachable. Both union codes in the examples above have a minimum Hamming distance of three, which enables single-bit error correcting capability at the first-tier decoder.

\subsection{Discussions}
Lemmas~\ref{lemma: HD_GB_U} and \ref{lemma: HD_KK_U} indicate that the minimum Hamming distance of the union code corresponding to a Gabidulin code or a KK code is one. Nevertheless,  the first-tier decoding can be carried out in three ways. First, as shown by Lemmas~\ref{lemma: HD_KK_cardi} and \ref{lemma: HD_MV_cardi}, the union code corresponding to a KK code or an MV code is a proper subset of the ambient vector space. Hence,  when a received packet does not belong to the union code, it can be discarded. Hence, the first-tier decoder is nontrivial even if the union code has a minimum distance of one. Second, the two examples in Section~\ref{sec: HD_MV} show that the minimum Hamming distance of the union codes corresponding to MV codes can be larger than one. In such cases, the union code not only allows the first-tier decoder to reject packets that do not belong to the union code, but also enables correction of single-bit errors in any received packet. Third, so far we have focused on the Hamming distance properties of the union code. In a sense, this is the worst-case scenario: if the first-tier decoder has no information about the transmitted packets, it has to assume any vector in the union code is likely. However, if the second-tier decoder is a list decoder (such as those proposed by \cite{mahdavifar_it10}), the second-tier decoder outputs a list ${\mathcal L}$ of subspaces as possible transmitted subspaces. Given this information, all valid packets now constitute $\bigcup_{V_i \in {\mathcal L}}  C_i$. Since $\bigcup_{V_i \in {\mathcal L}} { C}_i \subseteq \bigcup_{V_i} {\mathcal{C_U}}$, the minimum Hamming distance of $\bigcup_{V_i \in {\mathcal L}} { C}_i$ is no worse than that of the union code. Hence, given the list ${\mathcal L}$ from the second-tier decoder, the first-tier decoder can re-decode the received packets with respect to $\bigcup_{V_i \in {\mathcal L}} {C}_i$.

\section{Conclusions and Future Works}\label{sec:conclusions}
In this paper, we enhance the error correction capability of subspace and rank metric codes  by using a novel two-tier decoding scheme. While the decoding of subspace and rank metric codes serves a second-tier decoding, we propose to perform a first-tier decoding on the packet level by taking advantage of Hamming distance properties of subspace and rank metric codes. Furthermore, we investigate the Hamming distance properties of Gabidulin codes, KK codes, and MV codes.

Our future works will further investigate the Hamming distance properties of Gabidulin codes, KK codes, and MV codes. In particular, we will focus on the Hamming distance properties of $\bigcup_{V_i \in {\mathcal L}} { C}_i$ for some list $\mathcal L$ of codewords.

\bibliographystyle{IEEEtran}

\begin{thebibliography}{10}
\providecommand{\url}[1]{#1}
\csname url@samestyle\endcsname
\providecommand{\newblock}{\relax}
\providecommand{\bibinfo}[2]{#2}
\providecommand{\BIBentrySTDinterwordspacing}{\spaceskip=0pt\relax}
\providecommand{\BIBentryALTinterwordstretchfactor}{4}
\providecommand{\BIBentryALTinterwordspacing}{\spaceskip=\fontdimen2\font plus
\BIBentryALTinterwordstretchfactor\fontdimen3\font minus
  \fontdimen4\font\relax}
\providecommand{\BIBforeignlanguage}[2]{{%
\expandafter\ifx\csname l@#1\endcsname\relax
\typeout{** WARNING: IEEEtran.bst: No hyphenation pattern has been}%
\typeout{** loaded for the language `#1'. Using the pattern for}%
\typeout{** the default language instead.}%
\else
\language=\csname l@#1\endcsname
\fi
#2}}
\providecommand{\BIBdecl}{\relax}
\BIBdecl

\bibitem{ahlswede_IT00}
{R. Ahlswede, N. Cai, S. Li, and R. Yeung}, ``Network information flow,''
  \emph{IEEE Trans. Info. Theory}, vol.~46, pp. 1204--1216, July 2000.

\bibitem{cai_itw02}
N.~Cai and R.~W. Yeung, ``Network coding and error correction,'' in \emph{Proc.
  IEEE Info. Theory Workshop}, October 2002, pp. 20--25.

\bibitem{HMKKESL_IT06}
{T. Ho, M. M\'edard, R. K\"{o}tter, D. Karger, M. Effros, J. Shi, and B.
  Leong}, ``{A random linear network coding approach to multicast},''
  \emph{IEEE Trans. Info. Theory}, vol.~52, no.~10, pp. 4413--4430, October
  2006.

\bibitem{KK_IT0808}
{R. K\"{o}tter and F. R. Kschischang}, ``{Coding for Errors and Erasures in
  Random Network Coding},'' \emph{IEEE Trans. Info. Theory}, vol.~54, no.~8,
  pp. 3579--3591, August 2008.

\bibitem{Delsarte}
{P. Delsarte}, ``{Bilinear Forms Over a Finite Field, with Applications to
  Coding Theory},'' \emph{J. Comb. Theory Ser. A}, vol.~25, pp. 226--241, 1978.

\bibitem{Gab85}
{E. M. Gabidulin}, ``{Theory of Codes with Maximum Rank Distance},''
  \emph{Problems on Information Transmission}, vol.~21, no.~1, pp. 1--12,
  January 1985.

\bibitem{Roth91}
{R. M. Roth}, ``{Maximum Rank Array Codes and Their Application to Crisscross
  Error Correction},'' \emph{IEEE Trans. Info. Theory}, vol.~37, pp. 328--336,
  March 1991.

\bibitem{SK_IT1209}
{D. Silva and F. Kschischang}, ``{On Metrics for Error Correction in Network
  Coding},'' \emph{IEEE Trans. Info. Theory}, vol.~55, no.~12, pp. 5479--5490,
  December 2009.

\bibitem{SKK_IT0908}
{D. Silva, F. R. Kschischang, and R. K\"{o}tter}, ``{A Rank-Metric Approach to
  Error Control in Random Network Coding},'' \emph{IEEE Trans. Info. Theory},
  vol.~54, no.~9, pp. 3951--3967, September 2008.

\bibitem{mahdavifar_isit10}
H.~Mahdavifar and A.~Vardy, ``Algebraic list-decoding on the operator
  channel,'' in \emph{Proc. IEEE Int. Symp. Info. Theory}, Austin, USA, June
  2010, pp. 1193--1197.

\bibitem{mahdavifar_it10}
------, ``{Algebraic List-Decoding on the Operator Channel},'' \emph{submitted
  to IEEE Trans. Info. Theory}, April 2010.

\end{thebibliography}

\end{document}